\documentclass[lettersize,journal]{IEEEtran}

\usepackage{cite}
\usepackage{amsmath,amssymb,amsfonts,amsthm}
\usepackage{mathrsfs}
\usepackage{array}
\usepackage{tabularx}
\usepackage{multirow}
\usepackage{adjustbox}
\usepackage{lipsum}
\usepackage{booktabs}
\usepackage{color,colortbl}
\definecolor{lavender}{rgb}{0.9, 0.9, 0.98}
\usepackage{mathtools}
\usepackage{subfigure}
\usepackage{mathtools}
\usepackage{graphicx}
\usepackage{bbm}
\usepackage[utf8]{inputenc}
\usepackage{soul}

\usepackage{pifont}

\usepackage{algpseudocode}
\usepackage[linesnumbered,ruled,vlined]{algorithm2e}
\SetKwInput{KwInput}{Input}                
\SetKwInput{KwOutput}{Output}
\SetKwInput{KwInit}{Initialization}

\newtheorem{theorem}{Theorem}

\newtheorem{remark}{Remark}

\newtheorem{lemma}[theorem]{Lemma}
\usepackage{textcomp}

\newtheorem{mydef}{Definition}

\usepackage[normalem]{ulem}
\allowdisplaybreaks

\usepackage[ruled,vlined]{algorithm2e}
\SetKwInput{KwInput}{Input}           
\SetKwInput{KwOutput}{Output}
\SetKwInput{KwInit}{Initialization}
\begin{document}

\title{Capacity Maximization for RIS-assisted Multi-user MISO Communication Systems}
\author{M. S. S. Manasa, Kali Krishna Kota, Praful D. Mankar, and Harpreet S. Dhillon
\thanks{M. S. S. Manasa, K.K. Kota, and P. D. Mankar are with Signal Processing and Communication Research Center, IIIT Hyderabad, India. Email: mss.manasa@research.iiit.ac.in, kali.kota@research.iiit.ac.in,   praful.mankar@iiit.ac.in. H. S. Dhillon is with Wireless@VT, Department of ECE, Virginia Tech, Blacksburg, VA (Email: hdhillon@vt.edu). 
}. 
\thanks{}}

\maketitle

\begin{abstract}
We consider a multi-user multiple input single output (MU-MISO) system assisted by a reconfigurable intelligent surface (RIS). For such a system, we aim to optimally select the RIS phase shifts and precoding vectors for maximizing the effective rank of the weighted channel covariance matrix which in turn improves the channel capacity. For a low-complex transmitter design, we employ maximum ratio transmission (MRT) and minimum-mean square error (MMSE) precoding schemes along with water-filling algorithm-based power allocation. Further, we show that MRT and MMSE exhibit equivalent performance and become optimal when the channel effective rank is maximized by optimally configuring the RIS consisting of a large number of elements.
\end{abstract}

\vspace{-4mm}
\section{Introduction}
Reconfigurable Intelligent surface (RIS) has demonstrated potential to revolutionize the next-generation wireless communication systems due to its ability to alter the propagation environment. It is an array whose elements can be tuned to create a smart propagation environment to combat fading and improve the system performance in terms of capacity, coverage,  accuracy of sensing and localization, etc. 
For the details on the fundamentals of RIS along with the new challenges and research opportunities it brings, please refer to the survey and tutorial papers \cite{Liu_Yuanwei_2021_RIS_Survey, Emil_2022_signalprocessing, DiRenzo_2020_journal} and references therein.

For capacity maximization, the current literature predominantly focuses on optimizing performance metrics, such as sum rate, and system parameters, such as transmit power, etc. by jointly selecting the RIS phase shift matrix and precoding/beamforming vectors via application of diverse optimization methods including majorization-minimization (MM), semi-definite relaxation (SDR),  block coordinate descent (BCD), successive convex approximation (SCA), weighted mean square error (WMSE), etc.    
For example, \cite{Wang_2020_SNRmax} optimizes the RIS phase shift using SDR to maximize the signal-to-noise ratio (SNR) for the RIS-aided multiple input single output (MISO) system.
The authors of \cite{Guo_2020_WSRmax} maximize the weighted sum rate using stochastic SCA and non-convex BCD for RIS-aided multi-user (MU)-MISO system under perfect and imperfect channel state information.  
Further, \cite{Linsong_2021_minRatemax}  maximize the minimum rate for the RIS-aided MU-MISO system where the RIS phase shift and input covariance matrix are selected using alternating gradient descent method-based solutions. While there is a vast body of literature addressing similar problems, an alternative line of research on capacity improvement via maximizing the degrees of freedom (DoF), or equivalently rank maximization, of channel matrix (e.g. see \cite{Gesbert_2000_MIMOrankmax} for MIMO systems) remains unexplored for RIS-aided systems. 

It is noteworthy that the spatial multiplexing gain is higher for a full-rank channel when the received SNR is higher \cite{tse2005fundamentals}. However,  high SNR conditions are often associated with line-of-sight (LoS) channels, which tend to exhibit a lower rank, consequently, this limits the multiplexing gains. The authors of \cite{Emil_2020_rankmaxRIS} demonstrated how RIS can be utilized to maximize the channel rank for $2\times 2$ MIMO systems. 
Nonetheless, a full-rank channel matrix is not enough to ensure maximum capacity unless it is well-conditioned \cite{tse2005fundamentals}. This gives rise to the need for enhancing the condition number (CN) of the channel matrix which can be characterized by a newly defined metric called \textit{effective rank} \cite{O.Roy_ER_introduced}. Inspired by this, the authors of \cite{Mossallamy_2020_basepaper} and \cite{Meng_2023_ER_simulation} employed effective rank metric for improving the DoF of the composite channel matrix (including direct and RIS-assisted indirect links) by configuring RIS. In particular,  \cite{Mossallamy_2020_basepaper} focused on configuring RIS phase shift to maximize the effective rank for orthogonalizing the columns of the channel matrix to enhance the performance of the minimum mean square error (MMSE) receiver. On the other hand, \cite{Meng_2023_ER_simulation} validates the benefits of effective rank maximization of channel matrix for RIS-aided system $2\times 2$ MIMO experimental prototype and simulation results. {\em However, how RIS phase shift and precoding matrices can be configured for maximizing the effective rank and their impact on capacity performance for RIS-aided MIMO systems still remains unaddressed.}

Motivated by this, this letter aims to maximize the effective rank of the weighted channel covariance matrix for improving the channel capacity for a RIS-aided MU-MISO system. In particular, we propose to jointly optimize the RIS phase shift and input covariance matrix to maximize the effective rank. The proposed solution configures RIS to maximize the effective rank using the gradient-descent approach and selects the precoding vector using linear MRT and MMSE schemes. Furthermore, we demonstrate that both maximum ratio transmission (MRT) and MMSE precoders yield equivalent performance and become optimal as the number of RIS elements increases sufficiently. 
As a result, employment of MRT in RIS-aided MU-MISO system ensures both optimal performance as well as low design complexity as it circumvents the need for complex optimization to obtain precoding vector. 


\section{System Model}
\label{system_model}
\begin{figure}[!ht]
\centering
  \centering
  \includegraphics[width=0.45\textwidth]{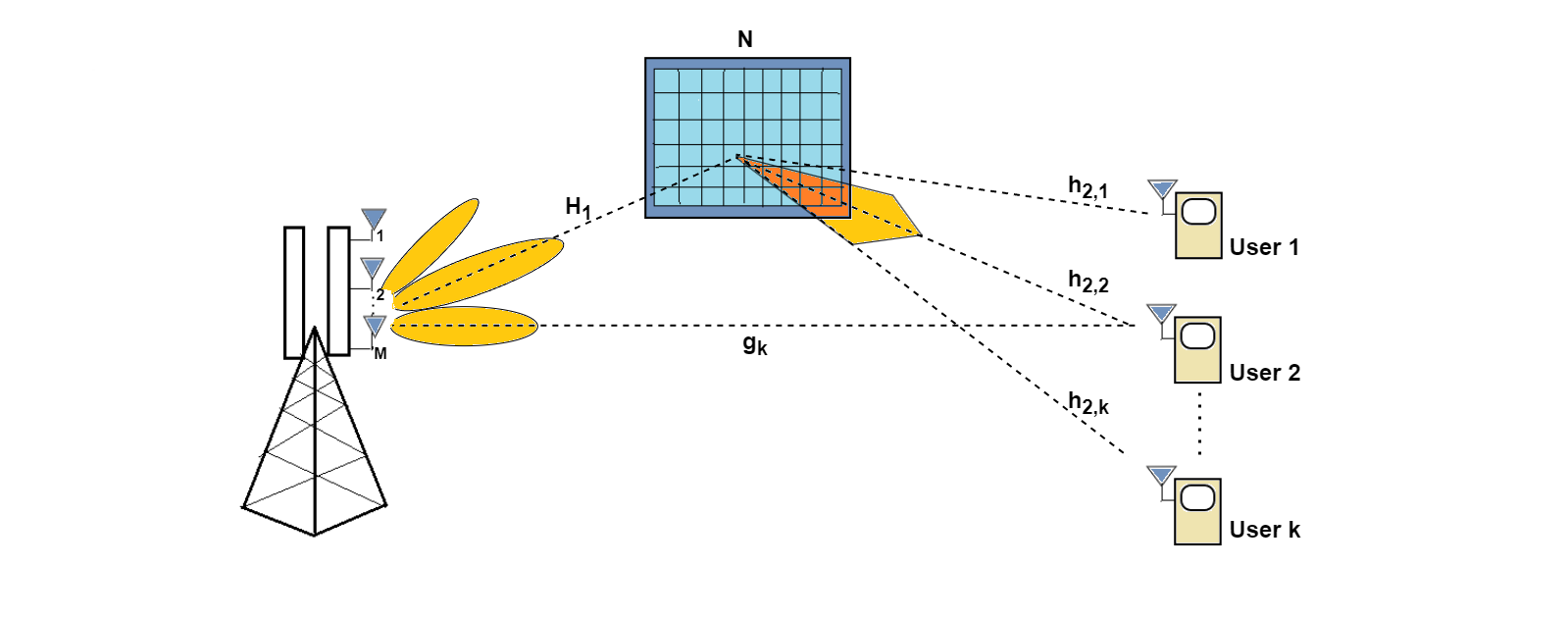}\vspace{-4mm}
\caption{Illustration of RIS-aided MU-MISO Downlink System.}\vspace{-2mm}
\label{Fig:SysModel}
\end{figure}

This paper considers a RIS-aided MU-MISO downlink communication system consisting of a base station (BS) equipped with $M$ antennas, $K$ single antenna users, and a RIS. The RIS is assumed to have $N$ passive elements whose phase shift matrix is  $\mathbf{\Phi} = {\rm diag} [e^{j\theta_1},  e^{j\theta_2},  \ldots e^{j\theta_N}]$ such that $\theta_j$ represents the phase shift provided by the $j$-th element. Let the baseband equivalent channels from the BS to the RIS be $\mathbf{H}_1 \in \mathbb{C}^{N \times M}$, and from the RIS to the $k$-th user be $\mathbf{h}_{2,k} \in \mathbb{C}^{1 \times N}$. The direct link between the BS to the $k$-th user is $\mathbf{g}_k \in \mathbb{C}^{1 \times M}$. The received signal at the $k$-th user can be written as
\begin{equation}
    y_k =   (\mathbf{g}_k + \mathbf{h}_{2,k}\mathbf{\Phi H}_1)\mathbf{x}_k + n_k, \nonumber
\end{equation}
where $\mathbf{x}_k = \sqrt{p}_k\mathbf{v}_ks_k$ such that $\mathbf{x}_k \in \mathbb{C}^{M \times 1}$ is the transmitted signal vector, $s_k$ is the message symbol, $\mathbf{v}_k$ is the unit norm precoding vector, and $p_k$ is the power associated with the $k^{th}$ user, and $\mathbf{n}_k \in \mathbb{C}^{K \times 1}$ represents the zero mean complex Guassian noise vector with variance $\sigma_n^2$. Without loss of generality, we assume $\mathrm{E}[ss^H] = 1$. The received signal in terms of the  composite channel matrix $\mathbf{H}$ can be written as
\begin{equation}
    \mathbf{y} = \mathbf{Hx} + \mathbf{n},
\end{equation}
where $\mathbf{x} = [x_1,\ldots,x_k]$, $\mathbf{H} = \mathbf{G} + \mathbf{H}_2\mathbf{\Phi H}_1$ is the composite MIMO channel matrix such that $\mathbf{G} = [\mathbf{g}_1,\ldots,\mathbf{g}_k]$, and $\mathbf{H}_2 = [\mathbf{h}_{2,1},\ldots,\mathbf{h}_{2,k}]$ and $\mathbf{n} = [n_1,\ldots,n_k]$. Let $\mathbf{h}_k$ be the $k$-th column of the matrix $\mathbf{H}$. The capacity  can be expressed as
\begin{equation}
    \mathrm{C} = \log_2|\mathbf{I} + \mathbf{HR}_x\mathbf{H}^H|,\label{MIMO_Capacity}
\end{equation}
where $\mathbf{R}_x = \mathbb{E}[\mathbf{xx}^H] = \sum_{k} p_k\mathbf{v}_k\mathbf{v}_k^H$ is the input covariance matrix. In this paper, our objective is to optimally select the precoding vectors, i.e. $\mathbf{v}_k$, and RIS phase shifts, i.e $\mathbf{\theta}$, such that the system capacity given in \eqref{MIMO_Capacity} is maximized. Consequently, we formulate the capacity maximization problem as
\begin{subequations}
\begin{align}
\max_{\theta, \mathbf{R}_x}~~& \log_2|\mathbf{I} + \mathbf{HR}_x\mathbf{H}^H|,\\
~~& \textit{\emph{s.t}} \hspace{0.3cm} \mathrm{Tr}(\mathbf{R}_x) \leq P_t,\label{power_constraint}\\
\textit{\emph{and}} ~~&  0 \leq \mathbf{\theta}_k < 2\pi \hspace{0.2cm}\forall k = 1,\ldots,N,
\end{align}
\end{subequations}
\label{Capacity_objective}where \eqref{power_constraint} represents the transmission power constraint given that the total available power is $P_t$. In the following section, we present an equivalent problem to optimally select the RIS phase shifts and employ linear precoding and power allocation techniques for configuring $\mathbf{R}_x$ to maximize the capacity.


\section{Optimal RIS Configuration and Precoding}
\label{ER-C}

In general, it is known that the MIMO channel capacity is proportional to the spatial DoF, i.e. the number of independent spatial dimensions available for communication. The rank of the matrix $\mathbf{HR}_x\mathbf{H}^H$, or equivalently the rank of the composite channel matrix $\mathbf{H}$, represents the number of spatially separable channels that can be utilized to form parallel streams for simultaneous transmission. Thus, a lower-rank channel matrix limits the ability of the system to take advantage of spatial diversity, which results in a lower capacity. However, a higher-rank channel matrix is not always sufficient to achieve capacity. The channel matrix condition number provides a more accurate characterization of channel capacity. A well-conditioned matrix (i.e. whose all eigenvalues are equal) provides better DoF which in turn results in higher capacity. Thus, the matrix $\mathbf{HR}_x\mathbf{H}^H$ must be well conditioned to attain higher capacity. From the spectral theorem, it can be verified that the eigenvectors of a well-conditioned matrix must be orthogonal to each other. From the above discussion, it is clear that orthogonalizing the columns of $\mathbf{HR}_x\mathbf{H}^H$ reflects as equalizing the eigenvalues, thereby increasing the MIMO capacity.
In addition, the  above observation can also be interpreted as the fact that the capacity-achieving matrix must have a higher value of
\begin{equation}
    \left|\mathbf{I}+\mathbf{HR}_x\mathbf{H}^H\right| =  \prod \left(1+\lambda_i\right),\label{det_prod}
\end{equation}
where $\lambda_i$'s are the eigenvalues of $\mathbf{HR}_x\mathbf{H}^H$. From \eqref{det_prod}, it is quite apparent that  $\lambda_i$'s should have equal values to maximize the product and thereby achieve higher capacity. Thus, the MIMO channel capacity maximization problem can be conveniently reframed as a problem to orthogonalize the columns of the composite channel matrix. This is where the RIS phase shifts play a vital role. By appropriately selecting RIS phase shifts, it is possible to alter the multipath fading to some extent such that it transforms the composite channel matrix into an orthogonal matrix. Moreover, it is essential to consider precoding alongside power constraints to further leverage the advantages offered by the multi-antenna transmitter. The problem of finding the optimal phase shifts of RIS $\mathbf{\Phi}$ and the precoding matrix $\mathbf{R}_x$ that orthogonalize the composite channel can be written in the form of a feasibility problem as
\begin{subequations}
\begin{align}
    \text{find} ~~& \mathbf{\Phi , R}_x\\
    \text{s.t} \hspace{0.3cm} (\mathbf{HR}_x\mathbf{H}^H)_i~~&(\mathbf{HR}_x\mathbf{H}^H)_j^H   = 0\hspace{0.3cm}\forall i\neq j\\
     \textit{\emph{and}} \hspace{0.3cm} \mathrm{Tr}&(\mathbf{R}_x) \leq P_t.
\end{align}
\label{Reframed objective}
\end{subequations}
As discussed above, orthogonalizing the columns of a matrix is closely tied to maximizing its rank. However, attaining rank maximization can prove to be an impractical solution, as a full-rank matrix might still exhibit ill-conditioning if there exists a substantial variation between eigenvalues. As a solution, we suggest maximizing a real-valued extension of rank known as the effective rank, which is introduced in \cite{O.Roy_ER_introduced}. The effective rank measures the effective dimensionality of the matrix and provides a cleaner way  to characterize how well-conditioned the matrix is.  
\begin{mydef}
The effective rank of a matrix $\mathbf{A} \in \mathbb{C}^{M\times K}$ is
\begin{equation}
    \mathcal{E}(\mathbf{A}) = \exp\left(-\sum_{i=1}^{\min(M,K)}\frac{\lambda_{i}}{\Vert \lambda\Vert_{1}}\log\frac{\lambda_{i}}{\Vert \lambda\Vert_{1}}\right),\label{er_def}
\end{equation}
\label{ER}
where $\lambda_i$'s are the eigenvalues of $\mathbf{A}$ and $\Vert\cdot\Vert_{1}$ is the $l_1$ norm.
\end{mydef}
The effective rank $\mathcal{E}(\mathbf{A}) \leq \mathrm{Rank}(\mathbf{A})$. A greater effective rank indicates a more uniform distribution of singular values, resulting in increased orthogonality among the  columns of matrix. Inspired by this, we reformulate  \eqref{Reframed objective} as
\begin{subequations}
\begin{align}
~~& \max_{\mathbf{\theta}, \mathbf{R}_x}\hspace{0.2cm} \mathcal{E}(\mathbf{HR}_x\mathbf{H}^H),\\
~~&  \textit{\emph{s.t}} \hspace{0.3cm} \mathrm{Tr}(\mathbf{R}_x) \leq P_t,\\
\textit{\emph{and}} ~~&  0 \leq \mathbf{\theta}_k < 2\pi \hspace{0.2cm}\forall k = 1,\ldots,N.
\end{align}
\label{ER-C_objective}
\end{subequations}
In the forthcoming sections, we focus our attention on solving the optimization problem in \eqref{ER-C_objective}.
We present methods to find the optimal solutions for phase shifts of RIS $\mathbf{\theta}$ and precoding matrix $\mathbf{R}_x$ that maximize the effective rank and consequently enhance the capacity.
\subsection{Design of phase shifts}
For a given input covariance matrix $\mathbf{R}_x$, the optimization problem in \eqref{ER-C_objective} becomes
\begin{subequations}
\begin{align}
    ~~&\max_{\mathbf{\theta}}\hspace{0.2cm} \mathcal{E}(\mathbf{HR}_x\mathbf{H}^H)\\
    \textit{\emph{s.t}} ~~&  0 \leq \mathbf{\theta}_k < 2\pi \hspace{0.2cm}\forall k = 1,\ldots,N,%
\end{align}%
\label{sub_prob1}%
\end{subequations}%
where $\mathbf{\theta}=\left[\theta_1, \theta_2, \ldots, \theta_N\right]$. This problem can be efficiently solved by employing a gradient-based technique to find the optimal solution. Thus, we first derive the gradient of $\mathcal{E}(\mathbf{HR}_x\mathbf{H}^H)$ with respect to the RIS phase shifts $\theta$ in a similar manner as done in \cite{Mossallamy_2020_basepaper}  as follows
\begin{equation}
    \frac{\partial \mathcal{E}}{\partial \theta_n}=\left[\begin{array}{llll}
\frac{\partial \lambda_{1}}{\partial \theta_n} & \frac{\partial \lambda_{2}}{\partial \theta_{n}} & \ldots & \frac{\partial \lambda_{K}}{\partial \theta_{n}}
\end{array}\right]\left[\begin{array}{c}
\frac{\partial \mathcal{E}}{\partial \lambda_{1}} \\
\frac{\partial \mathcal{E}}{\partial \lambda_{2}} \\
\vdots \\
\frac{\partial \mathcal{E}}{\partial \lambda_{K}}
\end{array}\right].
\label{grad}
\end{equation}
The partial derivative of $\mathcal{E}$ w.r.t the $k$-th eigenvalue can be found by directly differentiating \eqref{er_def} as
\begin{equation}
    \frac{\partial \mathcal{E}}{\partial \lambda_{k}}=\sum_{j=1}^{K} \frac{-C_{j, k}}{\left\|\mathbf{HR}_x\mathbf{H}^H\right\|_{1}^{2}}\left(1+\ln \frac{\lambda_{j}}{\left\|\mathbf{HR}_x\mathbf{H}^H\right\|_{1}}\right) \mathcal{E}\left(\mathbf{HR}_x\mathbf{H}^H\right),
\label{pd_2}
\end{equation}
\begin{equation}
    \text{where}~~C_{j, k}= \begin{cases}\sum_{i \neq k} \lambda_{i} & \text { if } j=k, \\ -\lambda_{j} & \text { if } j \neq k .\end{cases}\nonumber
\end{equation}
The singular value decomposition (SVD) of $\mathbf{HR}_x\mathbf{H}^H=\mathbf{U}^{*} \mathbf{S V}$. Then,
the partial derivative of $k$-th eigenvalue with respect to the $n$-th phase shift is given as
\begin{equation}
    \frac{\partial \lambda_{k}}{\partial \theta_{n}}=\mathbf{u}_{k}^{*} \frac{\partial (\mathbf{HR}_x\mathbf{H}^H)}{\partial \theta_{n}} \mathbf{v}_{k},
\label{pd_1}
\end{equation}
where $\mathbf{u}_{k}$ and $\mathbf{v}_{k}$ are the $k$-th columns of $\mathbf{U}$ and $\mathbf{V}$, respectively. Next, using  product rule of differentiation, we can write 
\begin{equation}
    \frac{\partial (\mathbf{HR}_x\mathbf{H}^H)}{\partial \theta_{n}} = \frac{\partial \mathbf{H}}{\partial \theta_{n}}\mathbf{R}_x\mathbf{H}^H + \mathbf{HR}_x\frac{\partial \mathbf{H}^H}{\partial \theta_{n}},
\end{equation}
where
\begin{subequations}
\begin{align}
   \left[\frac{\partial \mathbf{H}}{\partial \theta_{n}}\right]_{k, m}=~~&\mathrm{e}^{\mathrm{i}\left(\theta_{n}+\frac{\pi}{2}\right)}[\mathbf{H}_2]_{k,n}[\mathbf{H}_1]_{n,m},\nonumber\\
    \textit{\emph{and}} \left[\frac{\partial \mathbf{H}^H}{\partial \theta_{n}}\right]_{m, k}=~~&\mathrm{e}^{\mathrm{i}\left(\theta_{n}+\frac{\pi}{2}\right)}[\mathbf{H}_1]_{m,n}[\mathbf{H}_2]_{n,k}.\nonumber 
\end{align}
\end{subequations}
Now, by substituting \eqref{pd_1} and \eqref{pd_2}  into \eqref{grad}, we can determine the gradient of the effective rank $\mathcal{E}$. Thus, by using this gradient, the phase shift can be updated as 
\begin{equation}
    \mathbf{\theta}_n^* = \mathbf{\theta}_n + \alpha\times \frac{\partial \mathcal{E}}{\partial \theta_n},\label{gd}
\end{equation}
where $\alpha$ is the learning rate, a small positive constant that determines the step size in the direction of the steepest increase of the gradient function.
\subsection{Design of input covariance matrix}
For a given RIS phase shift $\mathbf{\theta}$, the problem \eqref{ER-C_objective} becomes
\begin{subequations}
\begin{align}
~~& \max_{\mathbf{R}_x}\hspace{0.2cm} \mathcal{E}(\mathbf{HR}_x\mathbf{H}^H),\\
~~&  \textit{\emph{s.t}} \hspace{0.3cm} \mathrm{Tr}(\mathbf{R}_x) \leq P_t.
\end{align}
\label{sub_prob2}
\end{subequations}
Recall that 
\begin{equation}
    \mathbf{R}_x = \sum_{k} p_k\mathbf{v}_k\mathbf{v}_k^H.\label{precoding} 
\end{equation}
The matrix $\mathbf{R}_x$ is pivotal in shaping the transmitted signal based on the CSI. The basic idea is to distribute the available power $P_t$ among different users using water-filling (WF) along with appropriate precoding vectors $\mathbf{v}_k$ that can exploit the spatial DoF offered by the channel in a way that maximizes the effective rank. While there exist various approaches to determine $\mathbf{v}_k$, our primary emphasis lies in attaining the stationary solution through the utilization of linear precoding schemes MRT and MMSE. While these schemes are suboptimal in general, our original motivation behind choosing them was to avoid using complex optimization algorithms to find the optimal $\mathbf{v}_k$, thereby ensuring low complexity transmit design. However, shortly, we will also prove that the MRT precoder will essentially be the optimal choice for the precoding when the number of RIS elements are sufficiently large. In the following, we first present the MRT and MMSE precoding scheme and later analyze their interplay in view of the proposed channel effective rank-maximization.

\subsubsection{MRT}
This precoding selects
\begin{equation}
    \mathbf{v}_k = \frac{\mathbf{h}_k^H}{\Vert\mathbf{h}_k\Vert}.\label{MRT}
\end{equation}
Thus, the received signal at the $k$-th user becomes
\begin{equation}
    y_{k}=\sqrt{p_{k}}\left\|\mathbf{h}_{k}\right\| s_{k}+ \sum_{j \neq k}\sqrt{p_{j}} \mathbf{h}_{k} \frac{\mathbf{h}_{j}^H}{\Vert\mathbf{h}_{j}\Vert} s_{j} + n_{k}. \label{MRTi}
\end{equation}
It is worth recalling that we have orthogonalized the columns of $\mathbf{H R}_{x} \mathbf{H}^H$ by optimally selecting $\mathbf{\Phi}$ (please refer to the discussion presented at the beginning of the Section \ref{ER-C}). Consequently, in the ideal case, we can assert that $\mathbf{h}_{k}\mathbf{h}_{j}^H=0$. Thus, it is evident that the interference in \eqref{MRTi} gets completely removed  (or, minimized if the maximum effective rank is not achieved). In such ideal scenario, the received signal and the corresponding SNR  can be expressed as
\begin{align}
    y_{k} ~~& = \sqrt{p_{k}}\left\|\mathbf{h}_{k}\right\|s_{k}+n_{k},\\ 
    \text{and} \hspace{0.2cm} \Gamma_{\text {MRT},k} ~~& = \gamma_k\left\|\mathbf{h}_{k}\right\|^2,\label{SNR_MRT}
\end{align}
respectively, where $\gamma_k = \frac{p_k}{\sigma_n^2}$.

\subsubsection{MMSE}
This precoding selects
\begin{equation}
  \mathbf{v}_k = \frac{\left(\mathbf{H}^{H}\mathbf{H}+\frac{1}{\gamma}_k \mathbf{I}\right)^{-1} \mathbf{h}_k^{H}}{\Vert\left(\mathbf{H}^{H}\mathbf{H}+\frac{1}{\gamma}_k \mathbf{I}\right)^{-1} \mathbf{h}_k^{H}\Vert}.\label{MMSE} 
\end{equation}
Thus, the received signal at the $k$-th user becomes
    \begin{align}
    y_k = ~~&\sqrt{p}_k\mathbf{h}_k\frac{\left(\mathbf{H}^{H}\mathbf{H}+\frac{1}{\gamma_k} \mathbf{I}\right)^{-1}\mathbf{h}_k^H}{\Vert\left(\mathbf{H}^{H}\mathbf{H}+\frac{1}{\gamma_k} \mathbf{I}\right)^{-1}\mathbf{h}_k^H\Vert}s_k + \nonumber\\
    ~~&\sum_{j \neq k}\sqrt{p_{j}} \mathbf{h}_k\frac{\left(\mathbf{H}^{H}\mathbf{H}+\frac{1}{\gamma_j} \mathbf{I}\right)^{-1}\mathbf{h}_j^H}{\Vert\left(\mathbf{H}^{H}\mathbf{H}+\frac{1}{\gamma_j} \mathbf{I}\right)^{-1}\mathbf{h}_j^H\Vert} s_{j} + n_{k}. \label{MMSEi}
    \end{align}
It can be noted that the maximum value of the effective rank of the channel matrix $\mathbf{H}$ is equal to $\min(M, K)$ and is achieved when all its eigenvalues $\lambda$ are equal. This can be easily verified using \eqref{er_def}. Such an ideal case corresponds to the scenario wherein the composite channel matrix $\mathbf{H}$ is orthogonal so that its columns are equally strong, i.e. $\Vert\mathbf{h}_k\Vert^2 = \Vert\mathbf{h}_k\Vert^2 = \lambda$. From this, we obtain $\mathbf{H}^{H}\mathbf{H} = \lambda\mathbf{I}$ using which we rewrite \eqref{MMSEi} as
\begin{equation}
    y_{k} = \sqrt{p_{k}}\left\|\mathbf{h}_{k}\right\| s_{k}+ n_{k},
\label{MMSEr}
\end{equation}
and the corresponding SNR can be written as
\begin{equation}
     \Gamma_{\text {MMSE},k} = \gamma_k\left\|\mathbf{h}_{k}\right\|^2.\label{SNR_MMSE}
\end{equation}
From  \eqref{SNR_MRT} and \eqref{SNR_MMSE}, it becomes clear that $\Gamma_{\text {MMSE},k} = \Gamma_{\text {MRT},k}$.
This establishes the equivalence of MRT and MMSE precoding schemes when the effective rank attains its maximum value, as highlighted in the following lemma.

\begin{lemma}
\label{Lemma}
  MRT  scheme provides capacity equal to that of the MMSE  scheme when the effective rank of the weighted channel covariance matrix $\mathbf{HR}_x\mathbf{H}^H$ is maximized by optimally configuring RIS with infinitely large number of elements.
\end{lemma}
\begin{proof}
Note that the effective rank can attain a higher value with the increase in the number of RIS elements $N$. As a result, we can say that the effective rank maximized with respect to RIS phase shift will provide a perfectly well-conditioned matrix for an RIS with an infinitely large number of elements. Thus, based on the above discussion for the ideal scenario, we can say that the columns of the optimally configured composite channel matrix are equally strong i.e. $\Vert\mathbf{h}_k\Vert^2 = \Vert\mathbf{h}_k\Vert^2 = \lambda$ and orthogonal to each other i.e. $\mathbf{h}_k\mathbf{h}_j^H = 0$. Therefore, in this ideal case, we can deduce using \eqref{SNR_MRT} and \eqref{SNR_MMSE} that the capacity difference of MRT and MMSE goes to zero.
\end{proof}

\begin{remark}
It is worth noting the impact of the number of RIS elements $N$ on the MRT and MMSE precoding performances. From Lemma \ref{Lemma}, it is reasonable to say that for moderately higher N, the MRT scheme provides a capacity that is almost equal to that of MMSE. Besides, with the further increase in $N$, the received SNR becomes significantly larger at which the performance of both MMSE and MRT precoding schemes converge with each other. Nonetheless, we can say that employing the proposed effective rank maximization results in the same performance of MRT and MMSE. This observation serves as a clear indication that employing MRT along with the optimally configured RIS can reduce the design complexity at the transmitter (as MMSE necessitates the computation of matrix inverses).\label{Remark}  
\end{remark}

The optimization problem in \eqref{ER-C_objective} is tackled by solving the two sub-problems \eqref{sub_prob1} and \eqref{sub_prob2} using the proposed solution iteratively, and the details are summarized in Algorithm \ref{Alg1}.

\begin{algorithm}\label{Alg1}

\KwInput{$\mathbf{H}_1$, $\mathbf{H}_2$ }
\KwOutput{$\mathbf{\Phi}^*$, $\mathbf{R}_{x}^*$}
\KwInit{$\mathbf{\Phi}_0$, $\mathbf{v}_{k_0}$, $p_{k_0}$, $i = 1$}

Calculate $p_{k_i}$ using WF.\\
Calculate $\mathbf{v}_{k_i}$ using MRT/MMSE using \eqref{MRT}/\eqref{MMSE}.\\
Calculate $\mathbf{{R}_x}_i$ using \eqref{precoding}.\\
Calculate the matrix $\mathbf{HR}_x\mathbf{H}^H$ for this configuration of $\mathbf{\Phi}_{i-1}$ and $\mathbf{{R}_x}_i$.\\
Solve the problem in \eqref{ER-C_objective} using the gradient descent method as per \eqref{gd}.
$\mathbf{i} \gets \mathbf{i}+1$\\
Continue until $|\mathcal{E}(\mathbf{\theta}_{i},\mathbf{R}_{x_i})-\mathcal{E}(\mathbf{\theta}_{i-1},\mathbf{R}_{x_{i-1}})| \leq \gamma$

\caption{Optimal configuration of RIS and input covariance matrix}
\end{algorithm}
\vspace{-4mm}
\section{Numerical analysis}
This section presents the numerical analysis to validate the capacity performance of the proposed scheme based on maximizing the effective rank of the channel matrix. Unless otherwise specified, the parameters for numerical analysis are considered to be: the number of transmit antennas $M=4$, the number of users $K=3$, the number of RIS elements $N=8$, and the noise variance $\sigma_n^2$ =1 and the transmit power $P_t= 10$ W. It is to be noted that direct link is not considered in our simulation in order to clearly understand the impact of RIS. The spectral efficiency (SE) is computed by averaging over $10^3$ independent channel realizations.
\begin{figure}[!ht]
    \centering\vspace{-2mm}
    \includegraphics[width=.45\textwidth]{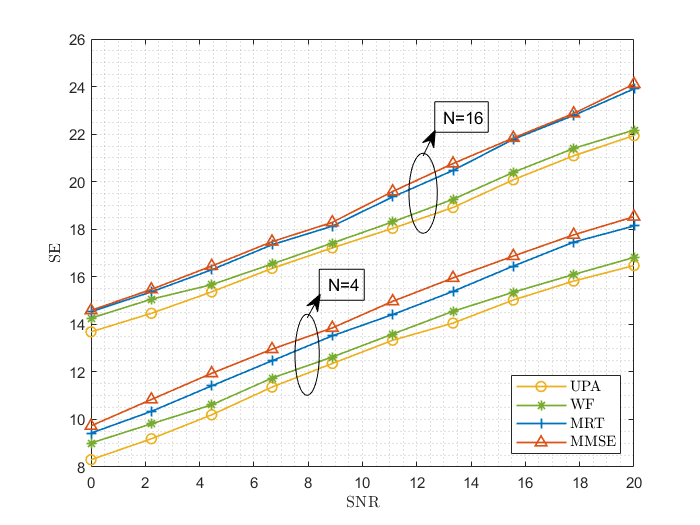}
    \vspace{-4mm}
    \caption{SE v/s SNR.}
    \vspace{-1mm}
\label{NR1}
\end{figure}

Figure \ref{NR1} shows the relationship between SE and SNR for the proposed algorithm. These results are provided for the cases wherein $\mathbf{\theta}$ is selected to maximize the effective rank (ER) and $\mathbf{R}_x$ is selected according to: Uniform power allocation (UPA), WF, MRT-WF, and MMSE-WF precoding schemes. It is to be noted that the effective-rank criteria in conjunction with precoding leads to a drastic improvement in the achievable SE. This is attributed to the fact that precoding optimizes the transmission signal according to the channel conditions, facilitating the transmission of interference-minimized parallel data streams. Consequently, precoding offers an additional advantage compared to UPA and WF schemes. Further, it can be seen that $N$ has a significant impact on SE which is expected. This result also shows that the performance gap of MRT and MMSE schemes reduces when $N$ increases from 4 to 16. This observation is also consistent with the result presented in Lemma \ref{Lemma}.
 
\begin{figure}[!ht]
    \centering
    \includegraphics[width=.45\textwidth]{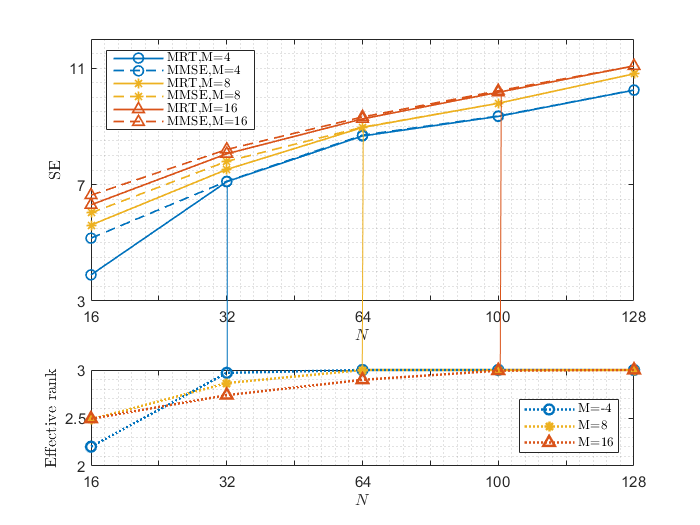}\vspace{-4mm}
    \caption{Top: SE vs. $N$. Bottom: Effective Rank vs. $N$. }\vspace{-1mm}
\label{NR2}
\end{figure}

Figure \ref{NR2} (top) shows that SE increases with the increase in either of the number of RIS elements $N$ and transmit antennas  $M$ or both for obvious reasons. The increase in SE  with respect to $N$ is gradual when $M$ is large.  However, the disparity in SE for different values of $M$ diminishes as $N$ becomes large. Further, it can be seen that the SE achievable by MRT converges to that of MMSE with the increasing $N$. This is because the effective rank (or equivalently condition number) of the composite channel matrix achievable by the proposed solution improves with the increase in $N$, which can be verified from Figure \ref{NR2} (bottom). Furthermore, note that this performance convergence occurs for the smaller values of $N$ when $M$ is small. This is attributed to the fact that, for small $M$, the dimension of the composite channel matrix is low for which the lesser number of RIS elements $N$ are sufficient for channel orthogonalization. 

\begin{figure}[!ht]
    \centering
    \includegraphics[width=.45\textwidth]{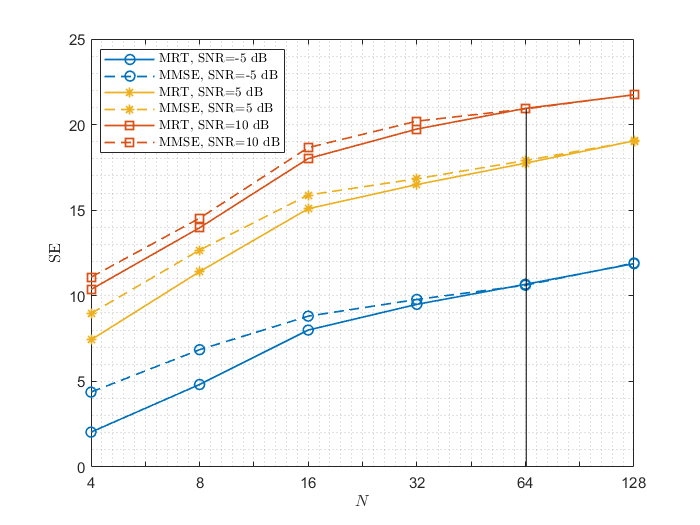}\vspace{-4mm}
    \caption{SE vs. $N$.}\vspace{-4mm}
\label{NR3}
\end{figure}

Figure \ref{NR3} shows the achievable SE with respect to $N$ for different SNRs. It can be seen that the performance convergence of MRT and MMSE occurs at almost the same value of $N$ regardless of the SNR. This is mainly because this value of $N$ is sufficient to orthogonalize the composite channel.\vspace{-2mm}

\section{Conclusion}\vspace{-1mm}
This letter presented a novel approach for maximizing the effective rank of the channel matrix for the RIS-aided MU-MISO system. The effective rank maximization improves the condition number of the channel matrix which in turn maximizes the channel capacity.  
In particular, we developed a gradient-descent approach to configure the RIS phase shift and employ the MRT/MMSE  precoding scheme along with WF-based power allocation to obtain the input covariance matrix. The proposed algorithm obtains a fixed point solution by solving the RIS phase shift and MRT/MMSE-WF-based input covariance matrix in an iterative manner. Further, we showed that MRT/MMSE along with optimally configured RIS perform equivalently and becomes optimal precoding when the number of RIS elements is large. Finally, we verify this observation through extensive numerical analysis for a wide range of SNRs and number of transmit antennas.
\vspace{-3mm}
\bibliographystyle{IEEEtran}
\bibliography{Ref}

\end{document}